\documentclass[preprint,12pt]{elsarticle}
\usepackage[left=2.0cm,right=2.0cm,top=1cm,bottom=2cm,bindingoffset=0cm]{geometry}
\usepackage{eepic}
\usepackage{amsmath,amssymb,amsthm}
\usepackage{multicol}
\usepackage{subfig}
\usepackage{enumitem}

\newtheorem{theorem}{Theorem}

\newtheorem{lemma}{Lemma}

\usepackage[T2A]{fontenc}
\usepackage[english]{babel}

\usepackage{listings, xcolor}

\begin{document}

\begin{frontmatter}

 \title{Dynamics of the Shapovalov mid-size firm model \footnote{This work was done under the auspices of the Institute for Nonlinear Dynamical Inference at the International Center for Emerging Markets Research (http://icemr.ru/institute-for-nonlinear-dynamical-inference/).} }
 
\author[hse]{Tatyana~A.~Alexeeva\,}
\author[kan,cefs]{William~A.~Barnett}
 \author[spbu,fin,ipm]{Nikolay~V.~Kuznetsov\,\corref{cor}}
 \ead{Corresponding author: nkuznetsov239@gmail.com}
 \author[spbu]{Timur~N.~Mokaev\,}
 
 \address[hse]{St. Petersburg School of Mathematics, Physics  and Computer Science, National Research University Higher School of Economics, 194100 St. Petersburg, Kantemirovskaya ul., 3, Russia}
 \address[kan]{Department of Economics, University of Kansas, Lawrence, KS 66045, USA}
 \address[cefs]{Center for Financial Stability, New York, NY 10036, USA}
 \address[spbu]{Faculty of Mathematics and Mechanics,
 St. Petersburg State University, 198504 Peterhof,
 St. Petersburg, Russia}
 \address[fin]{Department of Mathematical Information Technology,
 University of Jyv\"{a}skyl\"{a},  40014 Jyv\"{a}skyl\"{a}, Finland}
 \address[ipm]{Institute for Problems in Mechanical Engineering RAS, 199178 St. Petersburg, V.O., Bolshoj pr., 61, Russia}


\begin{abstract}

One of the main tasks in the study of financial and economic processes is forecasting and analysis of the dynamics of these processes. Within this task lie important research questions including how to determine the qualitative properties of the dynamics (stable, unstable, deterministic chaotic, and stochastic process) and how best to estimate quantitative indicators: dimension, entropy, and correlation characteristics.

These questions can be studied both empirically and theoretically. In the empirical approach, one considers the real data represented by time series, identifies patterns of their dynamics, and then forecasts short- and long-term behavior of the process. The second approach is based on postulating  the laws of dynamics for the process, deriving mathematical dynamic models based on these laws, and conducting subsequent analytical investigation of the dynamics generated by the models.

To implement these approaches, both numerical and analytical methods can be used. It should be noted that while numerical methods make it possible to study complex models, the possibility of obtaining reliable results using them is significantly limited due to calculations being performed only over finite-time intervals, numerical integration errors, and the unbounded space of possible initial data sets. In turn, analytical methods allow researchers to overcome these problems and to obtain exact qualitative and quantitative characteristics of the process dynamics. However, their effective applications are often limited to low-dimensional models (in the modern scientific literature on this subject, two-dimensional dynamic systems are the most often studied).

In this paper, we develop analytical methods for the study of deterministic dynamic systems based on the Lyapunov stability theory and on chaos theory. These methods make it possible not only to obtain analytical stability criteria and to estimate limiting behavior (localization of self-excited and hidden attractors, study of multistability), but also to overcome the difficulties related to implementing reliable numerical analysis of quantitative indicators (such as Lyapunov exponents and Lyapunov dimension). We demonstrate the effectiveness of the proposed methods using the ``mid-size firm'' model suggested recently by V.I.~Shapovalov as an example.
\end{abstract}


\end{frontmatter}

\section{Introduction}

Understanding and predicting the behavior of complex systems is one important task of current research in various fields. The events of the last decade have demonstrated the dangers of unpredictable developments in economic and financial systems, which can lead to systemic failures and even to collapses in the global financial-economic system. As a result, researchers have posed a number of conceptual questions, including search for approaches to forecasting critical transitions and determining stability indicators in complex systems \cite{Battiston_et.al.22-2016,Scheffer_etal21-2012}. Usually these characteristics are initially determined analytically, with subsequent experimental verification using real systems. At the same time, investigating and forecasting the dynamics of current financial and economic systems allows one to effectively control the systems, discovering and tracking stable, unstable, deterministic chaotic and stochastic processes. Studying current systems also allows quantitative estimation of the processes, using characteristics of their dimension, entropy and correlation.
To answer these conceptual questions, it is necessary to develop and apply reliable forecasting procedures, which will significantly reduce the costs of unpredictable behavior of financial and economic systems, including in times of crises, and allow researchers  to offer recommendations for stabilizing the dynamics of these systems.

In the second half of the last century, after the discovery of chaotic processes in dynamic systems by Ueda and Lorenz \cite{Lorenz-1963,UedaAH-1973}, chaos theory began to be actively developed, and helped to explain the complexity and unpredictability of dynamic systems behavior.
The complexity of dynamic systems, which are associated with the non-linearity and limited predictability of their behavior, as well as a number of open problems, have spurred significant interest in chaos theory among economists. Research aimed to study and reveal non-trivial economic effects and to stabilize irregular processes (see, e.g. \cite{BenhabibD8-1981,BenhabibN7-1979,BoldrinLM1-1986,BrockH-1998,BrockS14-1998,BrockHomEc-1997,BrockMBook3-1989,BullardB5-1993,Day2-1983,Day24-1992,FarmerOY-1983,SugiharaM4-1990}).
Developments in this field from the 1970s to the present can be traced through the studies of many famous economists \cite{BarnettSert10-2000,BarnettEr17-2013,BarnettUEr18-2016,BenhabibN7-1979,BrockH9-1998,BrockHomEc-1997,Day24-1992,GrandmontBook25-1998,Medio6-1992} who have explored numerous examples of deterministic economic models that could generate nonperiodic fluctuations.
For example, chaotic dynamics has been studied from the angles of economic growth and development, market structure and game theory \cite{Barkoulas26-2008,BoldrinLM1-1986,BrockMBook3-1989,Day2-1983}, rational expectations models \cite{BenhabibD8-1981,BrockH9-1998}, open economy New Keynesian models \cite{BarnettEr17-2013,BarnettUEr18-2016} based on the Gali and Monachelli model \cite{GaliM19-2005}, and non-linear heterogeneous agent models (HAMs) with periodic and chaotic asset price fluctuations \cite{BrockH9-1998,HommesBook12-2006}. Researchers have also focused on the issue of monetary chaos by applying tools based on both the metric (correlation dimension and Lyapunov exponents) and topological (recurrence plots) approaches to chaos \cite{Barkoulas26-2008}. There have been attempts to identify business cycles by looking at financial time series \cite{BrockS14-1998}, to analyze the stock returns of the markets of G7 countries \cite{TiwariaG15-2019} and to forecast high-frequency trading transactions in foreign exchange markets \cite{HirataA16-2012}.
More recently, research has developed in two main directions. An empirical direction is interested in whether some actual economic time series are characterized by chaotic dynamics and in developing statistical tests for chaos, applying them to macroeconomic and financial time series. 
A more theoretical approach has focused on demonstrating the possibility of cyclical and chaotic dynamic behaviors that can occur in a wide range of theoretical models. These studies attempt to understand whether mathematical non-linear deterministic models can demonstrate the types of fluctuations commonly found in economic data.
On the one hand, the theoretical direction has included analyses of low-dimensional dynamic models reconstructed from time series of real data \cite{Medio6-1992}. On the other hand, the theoretical approach is based on the construction of mathematical models that incorporated a-priori assumptions about the dynamics of the process under study.
The main goal of both approaches, empirical and theoretical, has been to build reliable forecasts of the behavior of dynamic systems based on concepts and methods of chaos theory.
The central complexity of the empirical approach lies in the difficulty of distinguishing between chaotic behavior of a deterministic dynamic system and random fluctuations caused by measurement or sampling errors \cite{BullardB5-1993,SugiharaM4-1990}.
The theoretical approach suffers from the fact that only low-dimensional systems allow derivations of analytical results (usually, two-dimensional systems are considered \cite{Barkoulas26-2008,BrockS14-1998,Medio6-1992}).

The history of the development of chaos theory and its applications shows that the complexity of describing the dynamics of economic phenomena is largely associated with the difficulty of constructing the adequate mathematical models of the phenomena.
Moreover, attempts to fully describe such models and to observe real data in order to reconstruct models from it typically lead to high-dimensional dynamic models, including stochastic ones.
As a rule, only quantitative analysis using numerical procedures is possible for these models.
It is well known that in chaos theory, confirming the reliability of results obtained by numerical methods requires separate clarification, including aspects connected with the use of shadowing theory and analysis of computational errors (caused by a finite precision arithmetic and numerical integration of differential equations) \cite{Pilyugin-2006}. 
Computational procedures have a number of significant limitations. First, calculations are performed over finite-time intervals, which makes it difficult to distinguish between a process of transition and established (limiting) chaotic behavior. Second, the use of numerical integration algorithms is associated with computational instability, which inevitably leads to approximation errors. Third, the theoretically unbounded space of possible initial conditions does not allow efficient forecasting of the limiting behavior of a dynamic system in the phase space.

At the same time, for low-dimensional systems, it is possible to apply rigorous nonlinear methods from the dynamic systems theory, which allows one to obtain exact analytical results. This enables a derivation of the effective criteria for stability and absence of chaotic behavior in such systems.
Thus, both qualitative exploration and estimation of the quantitative characteristics of the dynamics of the process can be performed numerically, as well as analytically.
In particular, it is possible to calculate the Lyapunov dimension of an attractor using numerical procedures and to analytically localize the attractor via determination of the absorbing set.

In this paper, we use an analytical approach to analyze local and global stability of the dynamics of a mid-size firm model \cite{ShapovalovKBA-2004,ShapovalKaz-2015}.
We perform analytical localization of the attractor of the system and investigate the global stability of its dynamics.
This allows us to obtain parameter domains in which the system demonstrates various types of behavior: stable, unstable, or deterministic chaotic dynamics.
In addition, we solve the problem of forecasting established (limiting) behavior of this dynamic system, obtaining the condition of the global conversion to the stationary set and bouded localization of non-trivial attractor.
Thus, we overcome the challenge of unboundedness of the set of initial data and implement reliable numerical analysis of the model, including the study of its chaotic dynamics.
We use the adaptive algorithm of the finite-time Lyapunov dimension and Lyapunov exponents computation for the values of the model's parameters at which a chaotic attractor can be confirmed. Thus, we obtain a number of quantitative estimates, which allow us to calculate characteristics including the Lyapunov dimension and entropy.

\section{Problem statement}

Consider the model of V.I.~Shapovalov proposed in \cite{ShapovalovKBA-2004} which describes behavior of a mid-size firm

\begin{equation}
\begin{cases}
\begin{aligned}
&\dot x=-\sigma x+\delta y,\\
&\dot y=\mu x +\mu y-\beta xz,\\
&\dot z=-\gamma z+\alpha xy.\end{aligned} 
\label{SysShap6} 
\end{cases}
\end{equation}
Here $\alpha$, \, $\beta$, \, $\sigma$, \, $\delta$, \, $\mu$, \, $\gamma$ are positive parameters, and the variables $x$, $y$, $z$ denote the growth of three main factors of production: the loan amount $x$, fixed capital $y$ and the number of employees $z$ (as an increase in human capital). An increase in the loan amount is proportional to the amount of capital and the size of the loan taken out. The coefficient with the variable $y$ is positive on the premise that, with an increase in capital, the company is more likely to grant loans on the lending market; the coefficient for the variable $x$ is negative and indicates the losses that the company incurs when taking out a new loan, which is associated with the requirement to pay interest, as well as the fact that the company is less willing to give credit when it has many loan obligations. The capital gain is proportional to the income from the investment of available capital and the loan taken, as well as expenses for labor remuneration and loan repayment. The coefficients for the sum of the variables $x$ and $y$ are positive, since they show a positive effect of investing  in the development of production; the coefficient for the product of the variables $x$ and $z$ is negative, since it indicates the costs of the company. The increase in the number of employees is proportional to the capital, the loan taken and the current number of employees. The coefficient for the product of the variables $x$ and $y$ is positive, based on the assumption that the company may spend part of the amount of capital and the loan taken on attracting additional employees. A negative coefficient for the variable $z$ indicates that the outflow from the current number of employees due to dismissal or on their own initiative should be taken into account.

Coefficients at variables are control parameters: $\alpha$ reflects a combination of factors that contribute to creating a company image that will be attractive to new employees; $\beta$ summarizes factors that influence cost allocation; $\mu$ describes the effectiveness of capital investments (the effects of various taxes should be taken into account); $\gamma$ summarizes factors related to difficulties obtaining a loan; for example, a high interest rate, etc.

As part of the study of system \eqref{SysShap6}, \cite{GurinaD-2010,ShapovalovKBA-2004,ShapovalKaz-2015} formulated the task of nonlinear analysis of the system and its limit dynamics in order to predict the stability of the Shapovalov model \eqref{SysShap6} and determine the conditions under which the system has some predictable dynamics (the Shapovalov problem of a mid-size firm dynamics forecasting). The non-triviality of this problem lies in the fact that the system has an unstable state of equilibrium and may exhibit of chaotic dynamics.

\section{System transformation}

A significant number of papers has studied the behavior of three dimensional nonlinear dynamic systems. It is important to verify (see, e.g. \cite{LeonovK-2015-AMC}) which known systems can be reduced to system \eqref{SysShap6} using linear coordinate transformation. System \eqref{SysShap6} can be reduced to a Lorenz-like system

\begin{equation}
\begin{cases}
\begin{aligned}
&\dot x= - c x + c y, \\
&\dot  y=r x + y - x z, \qquad \mbox{where} \,\, c = \frac{\sigma}{\mu} , r=\frac{\delta}{\sigma}, b = \frac{\gamma}{\mu},\\
&\dot z=-b z + xy, \end{aligned}
\label{SysLorenz3}
\end{cases}
\end{equation}
using the following coordinate transformation

\begin{equation}
(x, y, z) \rightarrow  \left(\frac{\mu}{\sqrt{\alpha \beta}} x,\, \frac{\mu \sigma}{\delta \sqrt{\alpha \beta}} y, \, \frac{\mu \sigma}{\delta \beta} z\right), \, t \rightarrow \frac{t}{\mu}. 
\label{TrfShLor} 
\end{equation}
System \eqref{SysLorenz3} differs from the classical Lorenz system \cite{Lorenz-1963} in the sign of the coefficient at $y$ in the second equation, which is 1 here, while in the Lorenz system this coefficient is -1.

Accordingly, the inverse transformation

\begin{equation}
(x, y, z) \rightarrow  \left(\frac{\sqrt{\alpha \beta}}{\mu} x, \frac{r \sqrt{\alpha \beta}}{\mu} y, \frac{r \beta}{\mu} z\right),\, t \rightarrow \mu t
\label{InvtrfLorSh} 
\end{equation}
reduces system \eqref{SysLorenz3} to system \eqref{SysShap6} with coefficients $\sigma = c \mu, \delta = r c \mu, \gamma = b \mu$ \footnote{Transformations \eqref{TrfShLor} and \eqref{InvtrfLorSh} do not change the direction of time, which is essential for the analysis of the Lyapunov exponents and dimension \cite{LeonovK-2015-AMC}.}.

In addition, system \eqref{SysShap6} with parameters satisfying the relations $\sigma^2 / (\sigma- \delta) = \mu $ and $ \delta <\sigma <\mu$ can be reduced to the well-known Chen system \cite{ChenU-1999}

\begin{equation}
\begin{cases}
\begin{aligned}
&\dot x= - a x + a y, \\
&\dot  y= (c - a) x + c y - x z, \qquad \mbox{with} \,\, a = \sigma , c=\frac{\sigma^2}{\sigma - \delta} = \mu, b = \gamma, \, a<c,\\
&\dot z= - b z + x y, \end{aligned}
\label{SysChen}
\end{cases}
\end{equation}
using coordinate substitutions

\begin{equation}
(x, y, z) \rightarrow  \left(\frac{1}{\sqrt{\alpha \beta}} x,\, \frac{\sigma}{\delta \sqrt{\alpha \beta}} y, \, \frac{\sigma}{\delta \beta} z\right). 
\label{Trf:ShapChen} 
\end{equation}

The possibility of reducing system \eqref{SysShap6}  to the Chen system \eqref{SysChen}  under the above conditions shows the complexity of studying the mid-size firm model. The Chen system demonstrates much more complex behavior in terms of constructing its absorbing set \cite{BarbozaC-2011} than the Lorenz system, and the problem of analytical calculation of the dimension of its attractor \cite{LeonovK-2015-AMC} is still opened.

\section{Sustainability analysis}
Further, we analyze the system \eqref {SysLorenz3} and apply the inverse transformation \eqref{InvtrfLorSh} to obtain conditions on the parameters of system \eqref{SysShap6}. To solve the Shapovalov problem, using the standard stability analysis of dynamic systems, we calculate the equilibria of system \eqref{SysShap6}.
System \eqref{SysShap6} always has three equilibria

\begin{equation}
O_1^{(1)}=(0, 0, 0), O_{2,3}^{(1)}=\left(\pm \sqrt{\frac{\gamma \mu (\sigma + \delta)}{\alpha \beta \sigma}} , \pm \sqrt{\frac{\gamma \mu \sigma (\sigma + \delta)}{\alpha \beta \delta^2}} ,   \frac{\mu (\sigma + \delta)}{\beta \delta} \right).
\label{StPointShap} 
\end{equation}
Accordingly, system \eqref{SysLorenz3} also has three equilibria

\begin{equation}
O_1^{(2)}=(0, 0, 0), \qquad O_{2,3}^{(2)}=\left(\pm \sqrt{  {b}(  {r}+1)} , \pm \sqrt{  {b} (  {r}+1)},   {r}+1 \right).
\label{StPointLorenz} 
\end{equation}

For the Jacobian matrix of system \eqref{SysShap6}

\begin{equation}
J=\begin{pmatrix} -  \sigma  & \delta  &0\\
 \mu - \beta z & \mu &- \beta x\\
\alpha y & \alpha x& -  \gamma \end{pmatrix}
\label{Jacob_Shap} 
\end{equation}
the characteristic polynomial $\det (J - I s)$ has form

\begin{equation}
\chi(s,x,y,z)=s^3+p_1^{(1)}(x,y,z) s^2+p_2^{(1)}(x,y,z) s +p_3^{(1)}(x,y,z),
\label{CharPolShap} 
\end{equation}
where

\begin{equation}
\begin{aligned}
&p_1^{(1)}(x,y,z)=  \sigma  +  \gamma - \mu, \\
&p_2^{(1)}(x,y,z)=  \sigma (\gamma - \mu) - \mu (\gamma + \delta) + \alpha^2 x^2 + \alpha \delta z, \\
&p_3^{(1)}(x,y,z)= - \gamma \mu (\sigma + \delta) + \alpha^2 \sigma x^2 + \alpha^2 \delta x y + \alpha \delta \gamma z. \end{aligned} 
\label{Coef:RH_Shap} 
\end{equation}

\begin{lemma}\label{lemma1}
The equilibrium state $O_1^{(1)} = (0, 0, 0)$ of system \eqref{SysShap6} is unstable for all parameter values.
\end{lemma}

\begin{proof}

Consider system \eqref{SysLorenz3} obtained from \eqref{SysShap6} by changing variables \eqref{TrfShLor}, with the Jacobian matrix

\begin{equation}
J=\begin{pmatrix} -  c &  c &0\\
 r- z & 1 &- x\\
y & x& -  b\end{pmatrix}. 
\label{JacLor} 
\end{equation}
At the point $O_1^{(2)} = (0, 0, 0)$, the coefficients of the characteristic polynomial of the Jacobian matrix \eqref{JacLor} have the following form

\begin{equation}
\begin{aligned}
&p_1^{(2)}(0, 0, 0)=  c +  b - 1, \\
&p_2^{(2)}(0, 0, 0)=  c(  b- r) - (  c +  b), \\
&p_3^{(2)}(0, 0, 0)=-  c   b(1+ r). \end{aligned} 
\label{CoefRH_Lor 0} 
\end{equation}
Since inequality $p_3^{(2)} (0, 0, 0) <0$ holds for any admissible values of the parameters of system \eqref{SysLorenz3}, the Routh-Hurwitz conditions are not satisfied and the equilibrium $O_1^{(2)}$ is always unstable. Using \eqref{InvtrfLorSh} we obtain the statement of Lemma for system \eqref{SysShap6}.
\end{proof}

\begin{lemma}\label{lemma2}
If one of the relations

\begin{equation}
\left[
  \begin{aligned}
 &r > \frac{  c(3-(  c+  b))}{  b -(  c+1)}, \qquad b >   c +1,\\
 &r < \frac{  c(3-(  c+  b))}{  b -(  c+1)}, \qquad 3 -  c <   b <   c +1 
  \end{aligned}
\right. 
\label{CondStab1} 
\end{equation}
holds for system \eqref{SysLorenz3} then the equilibria $O_{2,3}^{(1)}$ of system \eqref{SysShap6} are stable.

If both relations \eqref{CondStab1} are not satisfied, then the equilibria $O_ {2,3}^{(1)}$ of system \eqref{SysShap6} are unstable.
\end{lemma}
\begin{proof}

Similar to Lemma~\ref{lemma1} we consider system \eqref{SysLorenz3} obtained from \eqref{SysShap6} by changing variables \eqref{TrfShLor}.
At points $O_ {2,3}^{(2)} = \left(\pm \sqrt{{b} ({r} +1)}, \pm \sqrt{{b} ({r} +1 )}, {r} +1 \right)$ the coefficients of the characteristic polynomial of the Jacobian matrix of system \eqref{SysLorenz3} are the following

\begin{equation}
\begin{aligned}
&p_1^{(2)}(O_{2,3}^{(2)})=  c +  b - 1, \\
&p_2^{(2)}(O_{2,3}^{(2)})=  b(  c + r), \\
&p_3^{(2)}(O_{2,3}^{(2)})=2  c  b( r +1). \end{aligned} 
\label{CoefRHS12} 
\end{equation}
If $c + b> 1$ is true then $p_1^{(2)} (O_{2,3}^{(2)})> 0$, $p_2^{(2)} (O_{2,3}^{(2)})> 0$, and $p_3^{(2)} (O_{2,3}^{(2)})> 0$ hold for any positive values of $c, \, b, \, r$.

Using the second relation from \eqref{CondStab1} the condition

\begin{equation}
p_1^{(2)}(O_{2,3}^{(2)}) p_2^{(2)}(O_{2,3}^{(2)}) - p_3^{(2)}(O_{2,3}^{(2)}) =  c (  c +  b -3) +  r (  b - (  c +1)) >0 
\label{CondRH} 
\end{equation}
holds. 
Hence, we obtain the stability conditions \eqref{CondStab1} for the equilibria $O_{2,3}^{(2)}$ .
Using \eqref{InvtrfLorSh} we obtain the statement of Lemma for system \eqref{SysShap6}.
\end{proof}

\section{Analytical localization of the global attractor}

It is important to show that system \eqref{SysShap6}  does not have trajectories tending to infinity either for a finite or for an infinite period of time for a correct mathematical description of economic processes in a model and the possibility of studying its limit dynamics. Next, we distinguish the domain of the parameters of system \eqref{SysShap6} for which all trajectories are bounded and, moreover, which over time fall into a limited closed region called an \emph{absorbing set} \cite{BarbozaC-2011}. For the corresponding set of parameters system \eqref{SysShap6} has a global attractor.

Using the ideas presented in \cite{BoichenkoLR-2005, Leonov-2018-UMZh, Smith-1986}, we can prove the following


\begin{lemma}\label{lemma3}
If $\gamma < 2\sigma$, then for any solution of system \eqref{SysShap6}
we have the following estimate
\begin{equation}\label{abs_set:eqstim1}
\liminf_{t \to +\infty} \big[\, z(t) - \frac{\alpha}{2\delta}x^2(t)\,\big] \geq 0.
\end{equation}
\end{lemma}

\begin{proof}

For system \eqref{SysLorenz3} and the Lyapunov function
\[
    V(x,z) = z - \frac{x^2}{2 c}
\]
we have
\[
    \dot{V}(x(t),z(t)) = - b \, V(x(t),z(t)) + \left(1 - \frac{b}{2c}\right) x^2(t).
\]
If $b < 2 c$ then
\[
    V(x(t),z(t)) \geq \exp(-bt) \, V(x(0),z(0)),
\]
that yields estimate \eqref{abs_set:eqstim1}.
Thus, the global attractor is located in the positive invariant set
representing a parabolic cylinder (Fig.~\ref{fig:shapovalov:absset-all})
\begin{equation}\label{invariant_set1}
  \Omega_1 = \left\{(x,y,z) \in \mathbb{R}^3 ~|~  z \geq \frac{x^2}{2 c} \right\}.
\end{equation}

Using \eqref{InvtrfLorSh} we obtain the statement of Lemma for system \eqref{SysShap6}.
\end{proof}


\begin{theorem}\label{theorem:shapovalov:dissip}
If $\gamma > 2\mu$ and $\gamma < 2\sigma$, then all solutions of system \eqref{SysShap6} eventually fall into a bounded closed set.


\end{theorem}

\begin{proof}
For system \eqref{SysLorenz3} we consider the Lyapunov function
\begin{equation}\label{lyap_func_1}
 V(x,y,z) = \frac{1}{2} \left[ A x^2 - 2 \, B \,x \, y + y^2 + \big(z - \big(r + (A + B) c - B\big)\big)^2 \right],
\end{equation}
where $A$ and $B$ are arbitrary positive parameters.

{\it Case 1.} If $A > B^2$ then
\begin{equation}\label{lyap_func_inf}
V(x,y,z) = \frac{1}{2} \left[ A (x - \tfrac{B}{A} y)^2 + (1 - \tfrac{B^2}{A}) y^2
+ \big(z - \big(r + (A + B) c - B\big)\big)^2 \right] \to \infty
\end{equation}
as $|(x,y,z)| \to \infty$.

For an arbitrary solution $u(t) = (x(t),y(t),z(t))$
of system \eqref{SysLorenz3} by Lemma~\ref{lemma3} we have 
\begin{align*}
 \dot{V}(x,y,z) = &- (A c + B r) x^2 - (B c - 1) y^2
 - (b - 2 B c) z^2 + (r + (A + B) c - B)  b z - { 2 \,c \,B \,z \big(z - \tfrac{x^2}{2 c} \big)}\\
\leq&
 - (A c + B r) x^2 - (B c - 1) y^2
 - (b - 2 B c) z^2 + (r + (A + B) c - B)  b z .
\end{align*} 

In order to have $B c - 1 > 0$, $b - 2 B c > 0$,
we choose $B \in \big(\frac{1}{c}, \frac{b}{2 c}\big)$
under the assumptions $b > 2$, $ b < 2 c$.
Suppose that $\varepsilon \in \big(0, \, b - 2 B c\big)$ and
$\lambda = \min\big\{A c + B r, \, B c - 1, \, (b - 2 B c)-\varepsilon\big\} > 0$.
Then
\begin{align*}
\dot{V}(x,y,z) = & - (A c + B r) x^2 - (B c - 1) y^2
 - (b - 2 B c - \varepsilon) z^2 - \varepsilon z^2 + (r + (A + B) c - B) b z \\
 = & - (A c + B r) x^2 - (B c - 1) y^2 - (b - 2 B c - \varepsilon) z^2 -
\left(\sqrt{\varepsilon} z - \frac{(r + (A + B) c - B)  b}{2\sqrt{\varepsilon}}\right)^2 \\
& + \frac{(r + (A + B) c - B)^2 \, b^2}{4\varepsilon}
\leq - \lambda (x^2 + y^2 + z^2) + \frac{(r + (A + B) c - B)^2 \, b^2}{4\varepsilon}.
\end{align*}

Suppose that $x^2 + y^2 + z^2 \geq R^2$.
Then a positive $\varkappa$ exists such that
\[
  \dot{V}(x,y,z) \leq - \lambda R^2 + \frac{(r + (A + B) c - B)^2 \, b^2}{4\varepsilon} < -\varkappa \quad
  \text{for } \quad R^2 > \frac{1}{\lambda} \frac{(r + (A + B) c - B)^2 \, b^2}{4\varepsilon}.
\]
We choose $\eta > 0$ such that
\[
  \left\{(x,y,z)~|~ V(x,y,z) \leq \eta \right\} \supset
  \left\{(x,y,z)~|~ x^2 + y^2 + z^2 \leq R^2 \right\}.
\]
Thus, the relation $x^2 + y^2 + z^2 \leq R^2$ implies that
\[
\begin{aligned}
& A (x - \tfrac{B}{A} y)^2 + (1 - \tfrac{B^2}{A}) y^2 + \big(z - \big(r + (A + B) c - B\big)\big)^2 = \\
& A (x - \tfrac{B}{A} y)^2 + (1 - \tfrac{B^2}{A}) y^2 + z^2 - 2 \big(r + (A + B) c - B\big) z
+ \big(r + (A + B) c - B\big)^2 \leq 2\eta.
\end{aligned}
\]
According to the Cauchy-Bunyakovsky-Schwarz inequality, we have

\[
  A (x - \tfrac{B}{A} y)^2  ~ \leq ~ 2 A (x^2 + \tfrac{B^2}{A^2} y^2)  ~ \leq ~ 2 A (1 + \tfrac{B^2}{A^2}) R^2
\]
and, since
\[
  - 2 \big(r + (A + B) c - B\big) z ~ \leq ~ 2 \big(r + (A + B) c - B\big) |z| ~
  \leq ~ 2 \big(r + (A + B) c - B\big) R,
\]
it is sufficient to choose
\[
\eta \geq \frac{1}{2} \left[
  (2 A + 2 + \tfrac{B^2}{A}) R^2 +
  2 \big(r + (A + B) c - B\big) R + \big(r + (A + B) c - B\big)^2\right].
\]
If $R$ and $\eta$ are chosen as shown above, then system \eqref{SysLorenz3} has the following compact ellipsoidal absorbing set:
\[
  \mathcal{B}_0 = \left\{\!(x,y,z)\,\big|\,
  V(x,y,z) = \frac{1}{2} \left[ A x^2 - 2 \, B \,x \, y + y^2
    + \big(z - \big(r + (A + B) c - B\big)\big)^2 \right] \leq \eta \! \right\}.
\]

{\it Case 2.} If $A = B^2$ then, using technics from {\it Case~1}, we attain the attractor of dynamic system \eqref{SysShap6} is located in the positive invariant set representing an elliptic cylinder: 
 \[ \Omega_2 = \left\{\!(x,y,z)\,\big|\,
  V(x,y,z) = \frac{1}{2} \left[ B^2\big(x - \tfrac{1}{B} \, y\big)^2 \, 
    + \big(z - \big(r + B^2 c + B(c-1)\big)\big)^2 \right] \leq \eta \! \right\}.
\]

Condition \eqref{lyap_func_inf} holds for all $(x,\,y,\,z)$ except for line $x=\tfrac{1}{B}y$ which is symmetry axis of the elliptic cylinder $\Omega_2$.

Hence, we obtain a number of the following positive invariant sets: the absorbing set $\mathcal{B}=\Omega_1 \bigcap \mathcal{B}_0$ and the elliptic cylinder $\Omega_2$ (Fig.~\ref{fig:shapovalov:absset-all}).

\begin{figure}[ht]
  \centering
  \includegraphics[width=0.7\linewidth]{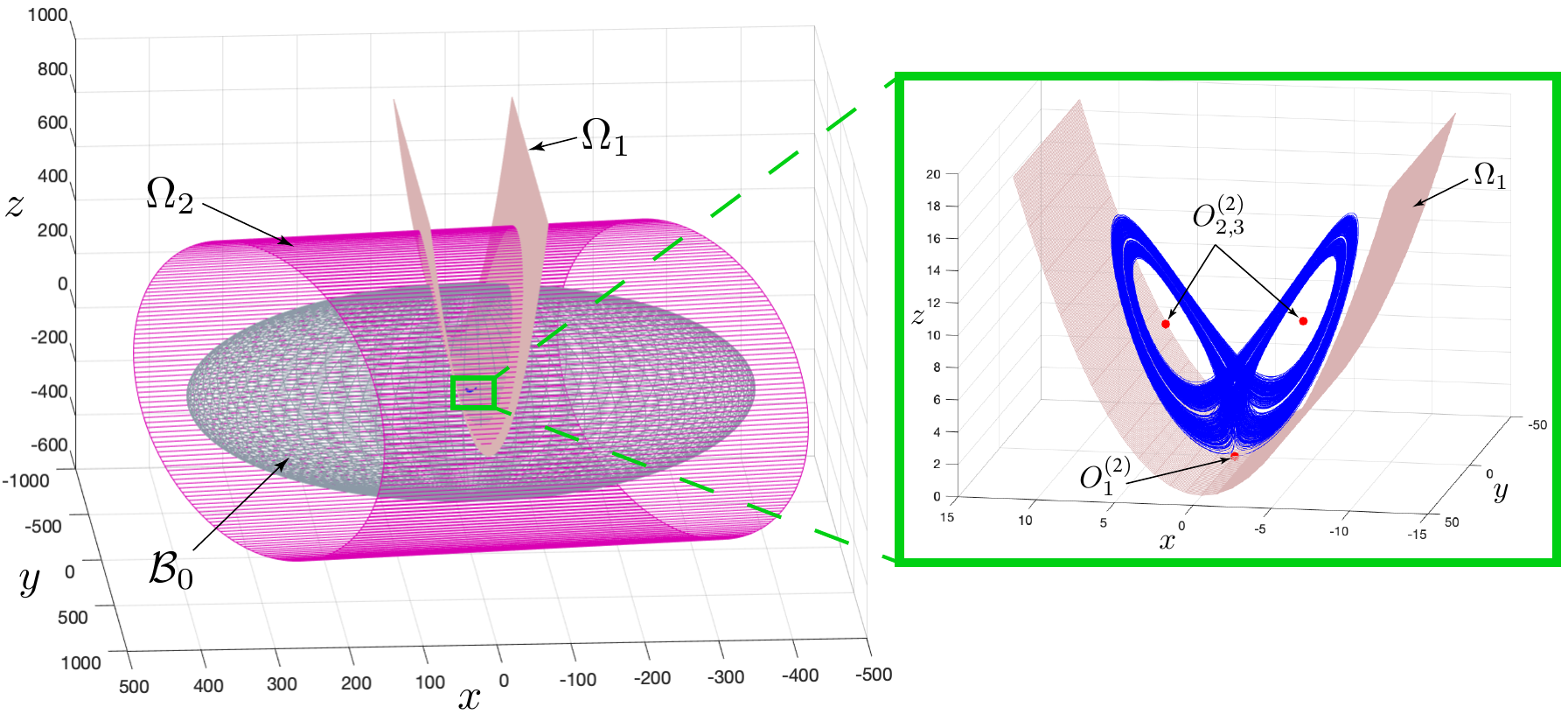}
  \caption{Localization of the chaotic attractor of system~\eqref{SysLorenz3}
  with parameters set at $r = 7.3171$, $b = 2.0909$, $c =3.7273$ by the
  absorbing set $\mathcal{B}=\Omega_1 \bigcap \mathcal{B}_0$, where the 
  $\Omega_1$ (parabolic cylinder, brown), $\Omega_2$ (elliptic cylinder, pink), $\mathcal{B}_0$ (ellipsoid, gray). Here $B=\tfrac{1}{2}\big(\tfrac{1}{c} + \tfrac{b}{2c}\big)=0.1167$, $A=B^2=0.0136$, $\varepsilon=\tfrac{1}{2}\big(b - 2 B c\big)=0.1667$, and $R$, $\eta$ are chosen as in the proof of Theorem~1.}  
  \label{fig:shapovalov:absset-all}
\end{figure}
Using \eqref{InvtrfLorSh} we obtain the statement of Theorem for system \eqref{SysShap6}.
Thus, system \eqref{SysShap6} generates a dynamic system, the solutions of system \eqref{SysShap6} exist for $t\in [0,+\infty)$ and system \eqref{SysShap6} possesses a global attractor, which contains all equilibria and nontrivial (local) attractors (see, e.g. \cite{Chueshov-2002-book,LeonovKM-2015-EPJST}).
\end{proof}

The absorbing set can be further refined following, e.g. \cite{liao2008absolute,ZhangFLiaoAll-2017}.

The problem of forecasting for established (limiting) behaviour of the dynamic system can be solved via localization of attractors of this system \cite{LeonovKM-2015-EPJST}. While trivial attractors (stable equilibrium points) can be easily found analytically or numerically, the search for periodic and chaotic attractors can be a challenging problem. For numerical localization of an attractor, one needs to choose an initial point in the basin of attraction and observe how the trajectory, starting from this initial point, visualizes the attractor after a transient process. Self-excited attractors, even those co-existing in the case of multistability \cite{PisarchikF-2014}, can be revealed numerically by the integration of trajectories, started in small neighbourhoods of unstable equilibria, while hidden attractors have the basins of attraction, which are not connected with equilibria and are ``hidden somewhere'' in the phase space \cite{DancaK-2017-CSF,Kuznetsov-2016,KuznetsovL-2014-IFACWC,LeonovK-2013-IJBC,LeonovKM-2015-EPJST}. 
Note that in numerical computation of a trajectory over a finite-time interval, it is difficult to distinguish sustained chaos from transient chaos (a transient chaotic set in the phase space, which can persist for a long time) \cite{GrebogiOY-1983}. 
For instance, for the classical Lorenz system the time interval for reliable computation with 16 significant digits and error $10^{-4}$ is estimated as [0,36], and reliable computation for a longer time interval, e.g. [0,10000] in \cite{LiaoW-2014}, is a challenging task that requires a significant increase in the precision of the floating-point representation and the use of supercomputers\footnote{Calculations by a supercomputer take about two weeks}.
Analytical aspects of this problem are concerned with the so-called shadowing theory (see, e.g. \cite{Pilyugin-2006}) which for some classes of systems can guarantee the existence of the ``true'' \,trajectory in the vicinity of its approximation.
Analytical procedures for determining the global stability areas are able to mitigate the influence of computer errors and thus make reliable forecasts of system dynamics.

\section{Global stability}

From the mathematical point of view, the problem when searching for the conditions for the  ultimate stationary behavior of the mid-size firm model \eqref{SysShap6} corresponds to the tendency of all its trajectories toward a stationary set. The nonlinearity of this system and the presence of one unstable equilibria $O_1(0,0,0)$ for all values of its parameters makes this problem nontrivial. Furthermore, for certain values of the parameters of the system this problem has a negative solution, and chaotic behavior is observed in system \eqref{SysShap6} \cite{ShapovalovKBA-2004}. We use specialized analytical methods to effectively obtain the analytical conditions for global stability of system \eqref{SysShap6} when all its trajectories tend to equilibria.

We write system \eqref{SysShap6} in general form

\begin{equation}
\dot u=f(u), \, f: \, U \subseteq \mathbb{R}^n \rightarrow \mathbb{R}^n ,
\label{SysTotal} 
\end{equation}
and consider its linearization along the solution $u(t,u_0)$ such that $u(0,u_0)=u_0 \in U$ exists for $t\in [0,\infty)$

\begin{equation}
\dot q=J(u(t,u_0))q,
\label{SysTotalLin} 
\end{equation}
where $f$ is continuously differentiable vector-function, $J(u)=Df(u)$ is the $n \times n$ Jacobian matrix, $\det J(u)\ne 0, \, \forall u\in U$. 

To check the global stability of system \eqref{SysTotal}, following \cite{Leonov-1991,Leonov-1991-Vest,Smith-1986}, it suffices to estimate the ordered eigenvalues of the symmetrized Jacobian matrix of this system

\begin{equation}
\frac{1}{2}\left(J(u(t,u_0))+J(u(t,u_0))^*\right) .
\label{SymJacobian} 
\end{equation}
For original nonlinear system \eqref{SysTotal}, computing the eigenvalues of the Jacobian matrix is technically difficult. Therefore, we use the linear coordinate transformation of system \eqref{SysTotal} with a nonsingular $n \times n$ matrix $S$

\begin{equation}
w=Su,
\label{TrfSlin} 
\end{equation}
with which it reduces to a transformed system

\begin{equation}
\dot w=Sf(S^{-1}w) .
\label{TrfSysTotal} 
\end{equation}
Its linearization along the solution $w(t,w_0)=S u(t,u_0)$ has the form

\begin{equation}
\dot v=S J(u(t,u_0)) S^{-1} v.
\label{SysTotalLinS} 
\end{equation}
The transformation \eqref{TrfSlin} often allows one to compactly write out and order the eigenvalues of the symmetrized Jacobian matrix of  system \eqref{TrfSysTotal}

\begin{equation}
\frac{1}{2}\left(SJS^{-1}+(SJS^{-1})^*\right)
\label{SymJacobianS} 
\end{equation}
and obviously does not affect system stability property \eqref{SysTotal}. 

 Thus, following \cite{Leonov-1991,Leonov-1991-Vest}, we can estimate the relation of the sum of the first two eigenvalues of the symmetrized matrix together with the derivative of the function $V(u)$. This allows us to effectively estimate the partial sum of the eigenvalues, while matrix $S$ is used to simplify calculation of the eigenvalues. Thus, we need the following relation

\begin{equation}
\lambda_1\left(u,S\right)+\lambda_2\left(u,S\right)+ \dot{V}\left(u\right)<0
\label{FormTh2} 
\end{equation}
for all $u$ from the absorbing set $\mathcal{B}=\Omega_1 \bigcap \mathcal{B}_0$.




\begin{theorem}\label{theorem:shapovalov:stability}
If for system \eqref{SysShap6} relations

\begin{equation}
\left\{
  \begin{aligned}
(  \gamma -\sigma )\left(  \frac{\gamma}{\mu} +1\right) <   \delta < ( \gamma + \sigma)\left(  \frac{\gamma}{\mu} -  1\right) , \\
    2 \mu < \gamma < 2 \sigma , \\
    \end{aligned}
\right. 
\label{CondStabFin_AbsSet1} 
\end{equation}
hold, then any solution tends to some equilibrium state for $t \to + \infty$.
\end{theorem}

\begin{proof}

Consider system \eqref{SysLorenz3} obtained from \eqref{SysShap6} by changing variables \eqref{TrfShLor}.
For system \eqref{SysLorenz3} the eigenvalues of Jacobian matrix \eqref{JacLor} have cumbersome expressions. So we apply the transformation \eqref{TrfSlin} with a nonsingular matrix $S$

\begin{equation}
 S=\begin{pmatrix}\frac{-1}{a} &0 &0\\
-\frac{  b  + 1}{  c} &1 &0\\
0& 0& 1\end{pmatrix}
\label{MatrixS} 
\end{equation}
to this system, where $a=\frac{c}{\sqrt{\left(1+  b\right)\left(  c-  b\right)+ r   c}}$, and condition 

\begin{equation}
r  c > (  b +1)(  b -  c) 
\label{CondA} 
\end{equation}
is satisfied.
Then the symmetrized Jacobian matrix of this system $\frac{1}{2}\left(SJS^{-1}+(SJS^{-1})^*\right)$ has the following eigenvalues

\begin{equation}
\lambda_2=-  b, \, \lambda_{1,3}=-\frac{  c -1}{2}\pm \frac{1}{2}\left( (2  b +1 -   c)^2 + a^{2}\left(\frac{  b +1}{  c}x+y\right)^2 + \left(a z-\frac{2  b}{a}\right)^2\right)^\frac{1}{2}.
\label{EigValMS} 
\end{equation}
The inequalities

\begin{equation}
2(\lambda_1 - \lambda_2) \ge 2 b +1- c  + | 2b + 1 - c|\ge 0
\end{equation}
imply $\lambda_1\ge\lambda_2$. From \eqref{EigValMS} we obtain the ratio

\begin{equation}
2(\lambda_1 + \lambda_2 )=-(  c -1 + 2  b)  + \left( (2  b +1 -   c)^2 + a^{2}\left(\frac{  b +1}{c}x+y\right)^2 + \left(a z-\frac{2  c}{a}\right)^2\right)^\frac{1}{2}. 
\label{SumLambda2} 
\end{equation}
Using the famous inequality $\sqrt{k+l} \le \sqrt{k} + \frac{l}{2\sqrt{k}}$, $\forall k>0, \, l\ge 0$, we obtain an estimate

\begin{equation}
\begin{aligned}
&2(\lambda_1 + \lambda_2)\le -(  c -1 + 2  b)  + \left[ (c+1)^2+4  c r\right]^{\frac{1}{2}}+ \\
&+\frac{2}{\left[ (  c+1)^2+4  c r\right]^\frac{1}{2}} 
\left[-  c z+\frac{a^2 z^2}{4} +\frac{a^2}{4}
\left(\frac{b + 1}{c}x+y\right)^2\right]. \end{aligned} 
\label{ESumLambda2} 
\end{equation}

We introduce the function $V(x,y,z)=\frac{\theta(x,y,z)}{\left[ (c+1)^2+4  c r\right]^{\frac{1}{2}}}$, where 

\begin{equation}
\theta(x, y, z) = h_1 x^2 + (- b h_2 + h_3) y^2 + h_3 z^2 + \frac{h_2}{4b} x^4 - h_2 x^2 z - h_2 h_4 xy -  \frac{c}{b}z,
\end{equation}
$h_j \, (j=\overline{1,4})$ are some positive real parameters. 
Then

\begin{equation}
\begin{aligned}
&2(\lambda_1 + \lambda_2)+2\dot V \le -(  c -1 + 2  b) - s(  c - 1) + (1-s)\left[ (  c+1)^2+4  c r\right]^{\frac{1}{2}}+\\
&+\frac{2(1-s)}{\left[ (  c+1)^2+4  c r\right]^\frac{1}{2}}\left[W(x,y,z)+\dot \theta\right],\end{aligned} 
\label{ESumLambda2V} 
\end{equation}
where $W(x,y,z)=-  c z+\frac{a^2 z^2}{4} +\frac{a^2}{4}\left(\frac{ b + 1}{c}x+y\right)^2$, and the function $\theta(x, y, z)$  is such that relation

\begin{equation}
W(x,y,z)+\dot \theta \le 0, \qquad \forall x,y,z\ge \frac{x^2}{2c} \, 
\label{Wtheta} 
\end{equation}
holds.
Using $\eqref{ESumLambda2}$, we obtain the following estimate

\begin{equation}
2(\lambda_1 + \lambda_2 + \dot V) \le -(  c -1 + 2  b) + \left[ (c +1)^2+4  c r\right]^{\frac{1}{2}}.
\label{ESumLambConst} 
\end{equation}
The relation $\lambda_1 + \lambda_2 + \dot V < 0$ is satisfied if for the parameters of system \eqref{SysLorenz3} the following condition

\begin{equation}
-(  c -1 + 2  b) + \left[ (  c+1)^2+4  c r\right]^{\frac{1}{2}} < 0
\label{CondStabF1} 
\end{equation}
holds.
This is the case if relations
\begin{equation}
\left\{
  \begin{aligned}
  c  r < (  b +  c)(  b -1), \\
    c +2  b >1, \\
    b > 1
  \end{aligned} 
\right. 
\label{CondStabF2} 
\end{equation}
are true.
From here, taking into account \eqref{CondA} and the condition of Theorem~\ref{theorem:shapovalov:dissip}, we obtain a domain of parameters

\begin{equation}
\left\{
  \begin{aligned}
(  b +1)\left(  \frac{b}{c} -1\right) <   r < \left(  \frac{b}{c} +  1\right)(  b -1) , \\
    2 < b < 2 c  \\
    \end{aligned}
\right. 
\label{CondStabFin_AbsSet2} 
\end{equation}
in which system \eqref{SysLorenz3} is globally stable.

Using \eqref{InvtrfLorSh}, we obtain the statement of Theorem for system \eqref{SysShap6}.
\end{proof}

Note that, according to the Fishing principle \cite{Leonov-2013-IJBC,LeonovKM-2015-EPJST,LeonovMKM-IJBC-2020} we obtain condition $3c>2b-1$ which is a necessary and sufficient condition for the existence of a homoclinic orbit (i.e. there exists a certain $r$ such that system \eqref{SysLorenz3} has a homoclinic orbit). This condition is valid under the second condition from \eqref{CondStabFin_AbsSet2}, and for the one can use a binary search \cite{LeonovMKM-IJBC-2020}.

\begin{figure}[ht]
  \centering
  \includegraphics[width=0.40\linewidth]{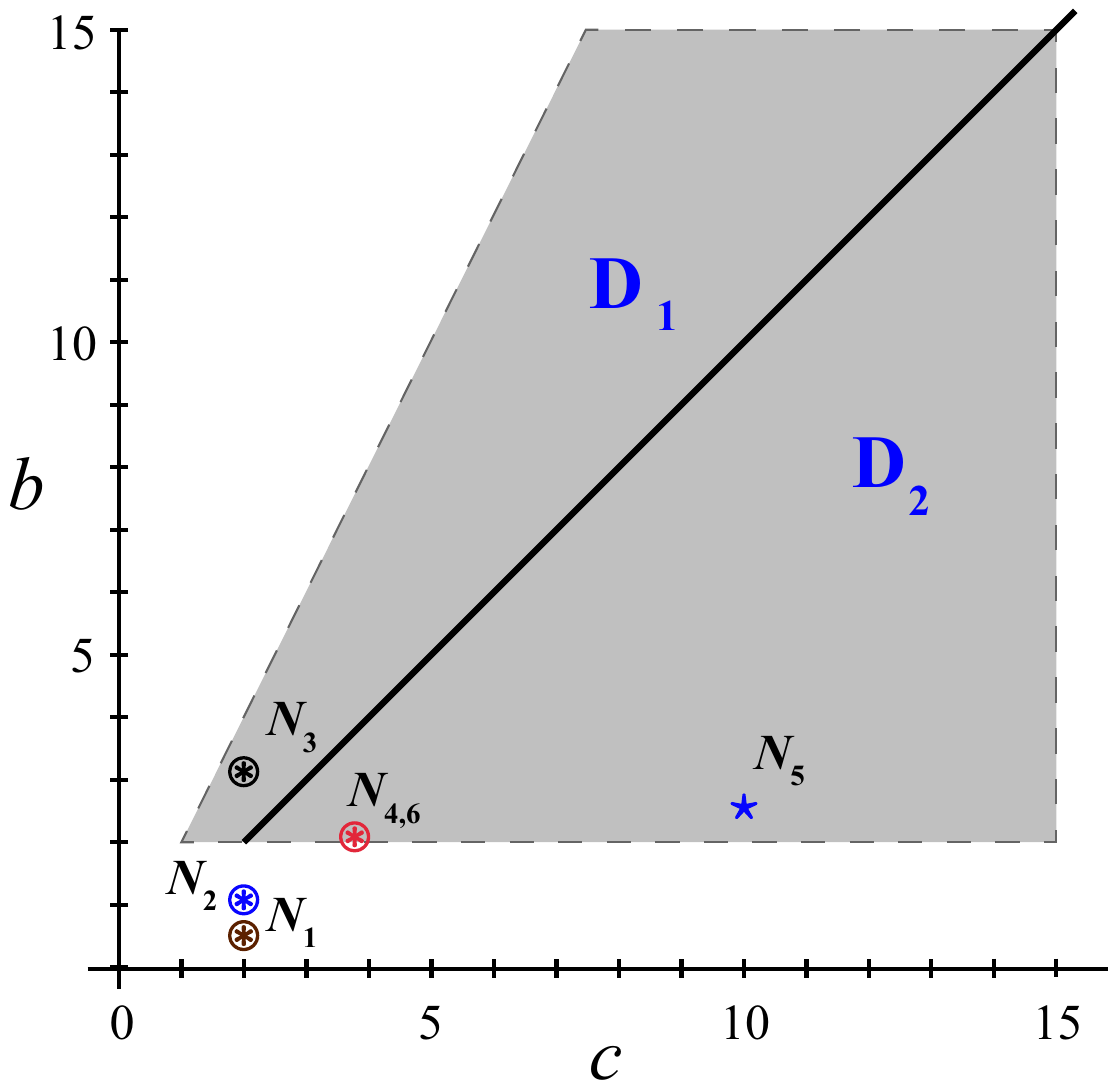}
  \caption{Parameters of system~\eqref{SysLorenz3} corresponding absorbing set $\mathcal{B}$
  and stability region \eqref{CondStabFin_AbsSet2}.}
  \label{fig:set}
\end{figure}

The Fig.~\ref{fig:set} shows the domain of values of parameters \eqref{CondStabFin_AbsSet2}, for which  system \eqref{SysLorenz3} has an absorbing set $\mathcal{B}$. At each point in this region, we can choose the value of the parameter $r$ such that if the first relation in \eqref{CondStabFin_AbsSet2} is satisfied, then system \eqref{SysLorenz3} and, accordingly, system \eqref{SysShap6}, is globally stable, otherwise, it is unstable.
The points $N_i = N_i (c, b, r), \, (i = \overline{1,5})$ correspond to the values of the parameters of system  \eqref{SysLorenz3}, at which the various behaviors are demonstrated.
So the point $N_1 (1.9524, \, 0.4762, \, 0.2439)$ corresponds to the classical Shapovalov's parameters \cite{ShapovalovKBA-2004}

\begin{equation}
\alpha=5, \, \beta=8, \, \gamma=1, \, \delta=1,\, \mu=2.1,\, \sigma=4.1,
\label{par:shapovalov} 
\end{equation}
and the point $N_2(1.9524, \, 1.1117, \, 0.2439)$ corresponds to parameters

\begin{equation}
\alpha=4, \, \beta=8, \, \gamma=2.3345, \, \delta=1,\, \mu=2.1,\, \sigma=4.1 
\label{par:gurina} 
\end{equation}
from \cite{GurinaD-2010}.
For the values of parameters \eqref{par:shapovalov} and \eqref{par:gurina} it is not possible to construct an absorbing set.
The points $N_3$, $N_4$, $N_5$ lie in the region of existence of the absorbing set. The point $N_3(2.1000, \, 3.1500, \, 4.0000) \in D_1$ corresponds to the stable regime, $D_1$ is the domain in which the inequalities $(b +1)\left(  \frac{b}{c} -1\right) <   r < \left(  \frac{b}{c} +  1\right)(b -1)$  are necessary to ensure the stability of  system \eqref{SysLorenz3}. The point $N_4(3.7273, \, 2.0909, \, 7.3171) \in D_2$ corresponds to the unstable regime, in domain $D_2$ the inequalities $(b +1)\left(  \frac{b}{c} -1\right) <   r < \left(  \frac{b}{c} +  1\right)(b -1)$ do not hold, because $r > \left(  \frac{b}{c} +  1\right)(b -1)$, that imply the instability of system \eqref{SysLorenz3}. The point $N_6(3.7273, \, 2.0909, \, 0.7317) \in D_2$ corresponds to the stable regime, in the domain $D_2$ the inequalities $0 <   r < \left(  \frac{b}{c} +  1\right)(b -1)$ imply the stability of system \eqref{SysLorenz3} (see Fig.~\ref{fig:grShapN6}). Note that for the classical parameters of the Lorenz system $c=10, \, b=8/3, \, r=24$ (the point  $N_5 \in D{}_2$) system \eqref{SysLorenz3} has the absorbing set and its equilibrium states are unstable. 

\begin{figure}[ht]
  \centering
  \includegraphics[width=0.5\linewidth]{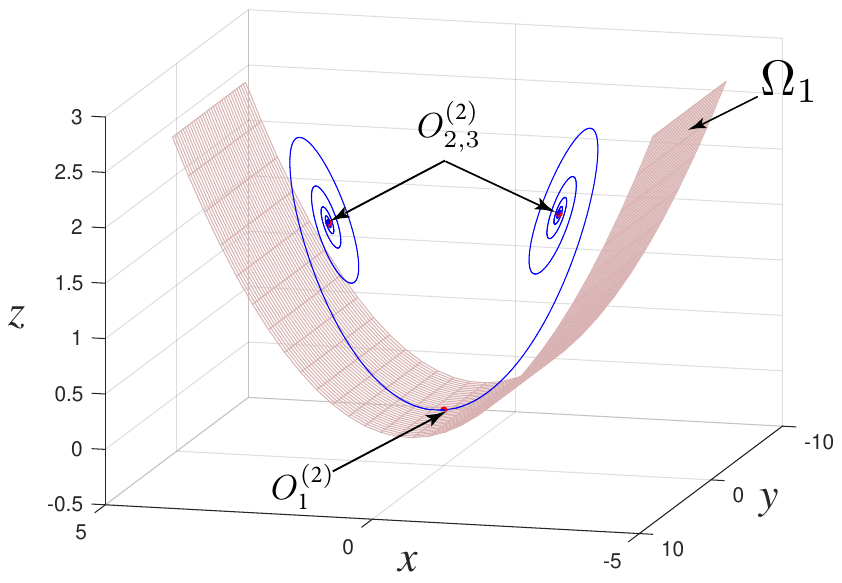}
  \caption{The stable regime of system~\eqref{SysLorenz3} 
  with parameters set at $r = 7.3171$, $b = 2.0909$, $c =0.7317$.}
  \label{fig:grShapN6}
\end{figure}

Thus, we showed that inside absorbing set $\mathcal{B}$ all trajectories of system  \eqref{SysLorenz3} not only enter in $\mathcal{B}$, but also tend to the stationary set defined by the relations \eqref{CondStabFin_AbsSet2}. A similar conclusion follows for system \eqref{SysShap6} by inverse transformation \eqref{InvtrfLorSh} to system \eqref{SysLorenz3}.

\section{Chaotic dynamics}

 Along with the problem of studying the stability of the mid-size firm model \eqref{SysShap6} formulated in \cite{ShapovalovKBA-2004,ShapovalKaz-2015}, as well the chaotic behavior of the model was analyzed. For example \cite{GurinaD-2010,ShapovalovKBA-2004,ShapovalKaz-2015}, the authors showed that this system exhibits chaotic behavior for the values of parameters \eqref{par:shapovalov}, which, taking into account the transformation \eqref{TrfShLor}, correspond to the following values of the parameters of system \eqref{SysLorenz3} $r = 0.2439$, $b = 0.4762$, $c = 1.9524$. Using numerical experiments, we have analyzed the chaotic dynamics of system \eqref{SysShap6}  and visualize a self-excited attractor for  Shapovalov's values of parameters \eqref{par:shapovalov}. On this attractor, along with the corresponding solution for system \eqref{SysShap6}, we have obtained some estimates, which allows us to calculate characteristics including the dimension of the attractor and entropy.

For a dynamic system $\varphi^t(u_0) = u(t,u_0)$
generated by system \eqref{SysShap6}
we consider the concept of the
\emph{finite-time Lyapunov dimension} \cite{Kuznetsov-2016-PLA,KuznetsovLMPS-2018},
which is convenient for carrying out numerical experiments with finite time:
\[
   \dim_{\rm L}(t, u) =
   j(t,u) + \tfrac{{\rm LE}_1\!(t,u) + \ldots + {\rm LE}_{j(t,u)}\!(t,u)
   }{|{\rm LE}_{j(t,u)\!+\!1}(t,u)|},
\]
where $j(t,u) = \max\{m: \sum_{i=1}^{m}{\rm LE}_i(t,u) \geq 0\}$, and
$\{{\rm LE}_i(t,u)\}_1^3$ is an ordered set of \emph{finite-time Lyapunov exponents} (FTLEs).
Then the \emph{finite-time Lyapunov dimension}
of a dynamic system generated by \eqref{SysShap6}
on compact invariant set $\mathcal{A}$ is defined as:
$\dim_{\rm L}(t, \mathcal{A}) = \sup\limits_{u \in \mathcal{A}} \dim_{\rm L}(t,u)$.
According to \emph{Douady--Oesterl\'{e} theorem},
for any fixed $t > 0$ the FTLD is an upper estimate of the Hausdorff dimension:
$\dim_{\rm H} \mathcal{A} \leq \dim_{\rm L}(t, \mathcal{A})$.
The best estimation is the \emph{Lyapunov dimension} \cite{Kuznetsov-2016-PLA}:
$\dim_{\rm L} \mathcal{A} = \inf_{t >0}\sup\limits_{u \in \mathcal{A}} \dim_{\rm L}(t,u)$.

In Fig.~\ref{fig:shapovalov:grid} shows the grid of points $\mathcal{C}_{\rm grid}$
filling the attractor:
the grid of points fills cuboid
$\mathcal{C} = [-1,1] \times [-2.5,2.5] \times [0.1,2.5]$
rotated by 45 degrees around the $z$-axis,
with the distance between points equal to $0.5$ (see Fig.~\ref{fig:shapovalov:grid}).
The time interval considered is $[0,\, T=500]$, $k=1000$, $\tau=0.5$,
and the integration method is MATLAB ode45 with predefined parameters.
The infimum on the time interval
is computed at the points $\{t_k\}_{1}^{N}$ at time step $\tau=t_{i+1}-t_i=0.5$.
Note that if, for a certain time, $t=t_k$ the computed trajectory is out of the cuboid,
the corresponding value of finite-time local Lyapunov dimension
is not taken into account in the computation of the maximum
of the finite-time local Lyapunov dimension
(e.g. if there are trajectories with initial conditions in cuboid,
which tend to infinity).

For the set of parameters considered, 
we use a MATLAB realization of {\it the adaptive algorithm of finite-time Lyapunov dimension
and Lyapunov exponents computation}~\cite{KuznetsovLMPS-2018}
and obtain the maximum of the finite-time local Lyapunov dimensions at the points of grid
($\displaystyle\max_{u \in \mathcal{C}_{\rm grid}} \dim_{\rm L}(t,u)$, at the time points $t=t_k=0.5\,k$ $(k=1,..,1000)$).
For parameters $r = 0.2439$, $b = 0.4762$, $c = 1.9524$
we get
$\max_{u \in \mathcal{C}_{\rm grid}} \dim_{\rm L}(100,u)=2.0699$ 
  $\max_{u \in \mathcal{C}_{\rm grid}} \dim_{\rm L}(500,u)=2.0676$.

\begin{figure}[ht]
  \centering
  \includegraphics[width=0.5\linewidth]{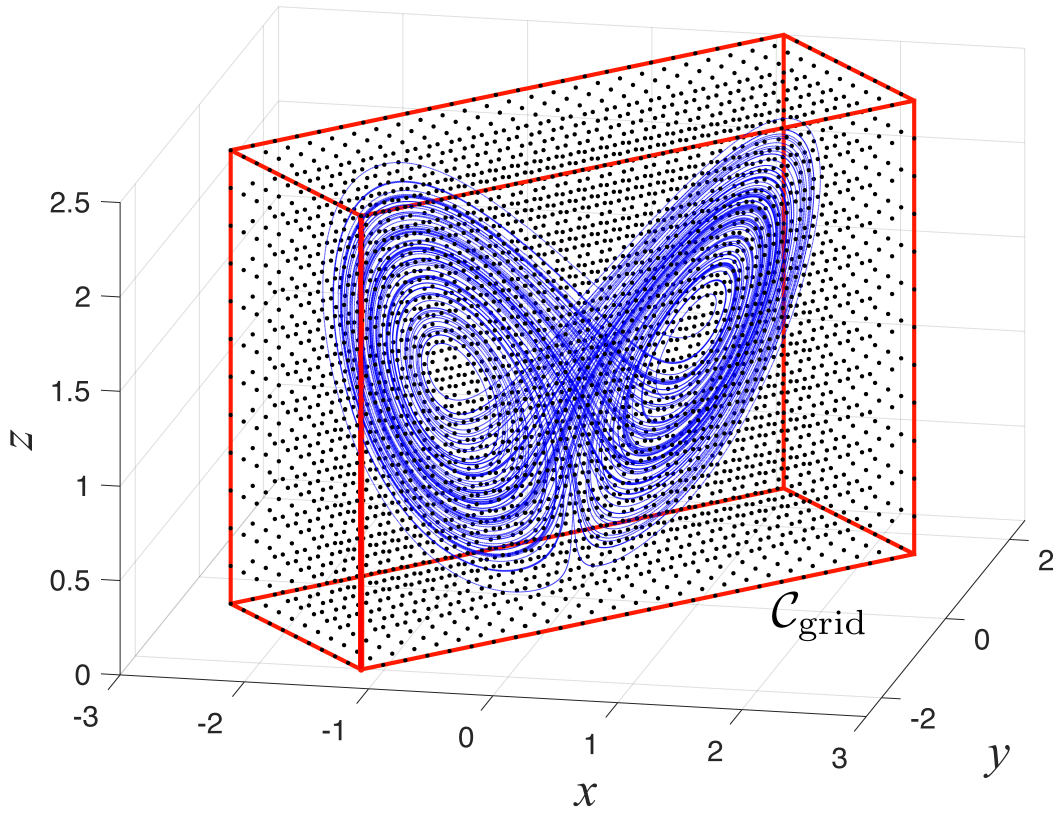}
  \caption{Localization of the chaotic attractor of system~\eqref{SysLorenz3}
  with parameters set at $r = 0.2439$, $b = 0.4762$, $c = 1.9524$ by the
  cuboid $\mathcal{C}$ and the
  corresponding grid of points $\mathcal{C}_{\rm grid}$.}
  \label{fig:shapovalov:grid}
\end{figure}

This estimation is consistent with the hypothesis on the Lyapunov dimension of a self-excited attractor. Here the difficulties of reliable estimation of the Lyapunov exponents and dimension along one randomly chosen trajectory over long-time interval are caused by unstable periodic orbits embedded into attractor and finite precision numerical integration of ODE (see \cite{AlexeevaBKM-2020-IFAC}). Note that the existence of multistability with hidden attractors in the model under consideration, as well as for the classical Lorenz system, is an open problem \cite{ChenKLM-2017-IJBC,KuznetsovM-2019-AIP,LeonovK-2015-AMC,LeonovKM-2015-EPJST,SprottM-2018-HA}.

\section*{Conclusion}

The complexity of analyses of the dynamics of financial and economic systems is often due to the presence of multistability, when, for different initial data, the system trajectories can converge to distinct attractors.
The coexistence of local attractors complicates forecasting the behavior of a dynamic system and estimating its quantitative characteristics, for instance, the Lyapunov dimension of the attractor.

Recent results obtained in the field of nonlinear methods of dynamic systems theory allow one to successfully analyze global and local dynamics using analytical procedures.
These studies have extensive applications including the analysis of global dynamics in a number of economic models \cite{BarnettBGMV-2020,LlibreV23-2018}.
One efficient method for analysis of these models is discovering global attractors by constructing their absorbing sets, and subsequently implementing effective analytical and numerical procedures for investigation bounded sets of initial conditions.

In this paper, we perform a global analysis of the stability of the mid-size firm model. We then derive the absorbing set, which makes it possible to localize the global attractor of system \eqref{SysShap6}, and the domains of the model parameters for which stability and instability are observed.
To characterize the chaotic dynamics of system \eqref{SysShap6}, we calculate the Lyapunov dimension of attractor for specific values of parameters.
Our work relies on an analytical approach to analyze the behavior of the model considered, which at the same time allows us to perform reliable numerical calculations of the Lyapunov dimension of the attractor of the mid-size firm model.

We believe that ongoing efforts in the field of forecasting dynamics of real processes should be focused on the development of efficient analytical and numerical procedures that allow researchers to obtain the most complete and reliable information about the dynamics of the processes.
This approach can help to expand the applicability of analytical procedures and overcome the disadvantages of numerical analysis. Hence, it opens a space for designing new control strategies for both stable regimes (including multistability, one of the most exciting phenomena in dynamic systems) and crisis processes.

\section*{Acknowledgments}
We dedicate this paper to the memory of Gennady A. Leonov (1947-2018), 
with whom we began this work in 2016.

We acknowledge support from the Russian Science Foundation (project 19-41-02002).



\end{document}